\documentclass[11pt, a4paper]{article}

\usepackage[round, longnamesfirst, sort&compress]{natbib}

\usepackage{geometry} 
\usepackage{fullpage}

\usepackage{amsmath}
\usepackage{amsthm}
\usepackage{amssymb}

\DeclareMathOperator*{\argmin}{arg\,min}
\newcommand{\pol}[1]{^{(#1)}}

\newcommand\numberthis{\addtocounter{equation}{1}\tag{\theequation}}

\begin{document}

\title{Gradient of the Objective Function for an Anisotropic Centroidal Voronoi
Tessellation (CVT) - A revised, detailed derivation}

\author{Giacomo Parigi\thanks{Corresponding author: giacomo.parigi@gmail.com},
Marco Piastra\\
 \small Computer Vision and Multimedia Lab,\\
 \small University of Pavia,\\
 \small Via Ferrata 1 - 27100 Pavia (PV), Italy}

\date{}

\maketitle

\abstract{In their recent article \citeyearpar{levy2010lpcvt}, Levy and Liu
introduced a generalization of Centroidal Voronoi Tessellation (CVT) - namely
the $L_p$-CVT - that allows the computation of an anisotropic CVT over a sound
mathematical framework. In this article a new objective function is
defined, and both this function and its gradient are derived in
closed-form for surfaces and volumes. This method opens a wide range of
possibilities, also described in the paper, such as quad-dominant surface
remeshing, hex-dominant volume meshing or fully-automated capturing of sharp
features. However, in the same paper, the derivations of the gradient and
of the new objective function are only partially expanded, in the appendices,
and some relevant requisites on the anisotropy field are left implicit. In order to better
harness the possibilities described there, in this work  the entire derivation process is made explicit.
In the authors' opinion, this also helps understanding the working conditions of
the method and its possible applications.

{\bf keywords:} Centroidal Voronoi tessellation, anisotropic meshing, surface
reconstruction, topology preservation, computational geometry and object
modeling}

\newtheorem{theorem}{Theorem}[section]

\section*{Introduction} In their recent article \citeyearpar{levy2010lpcvt}, Levy
and Liu introduce a generalization of Centroidal Voronoi Tessellation - namely
the $L_p$-CVT - that ``minimizes a higher-order moment of the coordinates on the
Voronoi cells''.  Levy and Liu take as reference the standard CVT objective function, 
and extend its behavior injecting the $L_p$ norm and an anisotropy term, in the form of a
matrix obtained from the anisotropy field, in the function itself. This method
opens a wide range of possible applications, also described in the article, such as
quad-dominant surface remeshing, hex-dominant volume meshing or fully-automated
capturing of sharp features. In particular, in the authors' opinion, this method
could also increase the resistance to noise in surface reconstruction and
remeshing, which is relevant in the application of methods such as the one described in 
\cite{piastra2013self}.

In the work by \cite{levy2010lpcvt}, however, the derivation of the gradient, as well as the
definition of the new objective function, are only partially described, and
some conditions on the anisotropy field are left implicit. Application of
the same method under different conditions, e.g. a specific anisotropy field, or for
different purposes, involves a complete comprehension of the mathematical frame on
which the method is based. This work is intended to analyze thoroughly the
derivation both of the objective function and of its gradient, in order to
understand the functioning of the method and the conditions of applicability.

For sake of clarity, the notation used here is slightly different from the one
in the original work, due to of the different structure of this paper.

\section{The Energy Function $F_{L_p}$}
\subsection{The objective function of $L_p$-Centroidal Voronoi Tessellation}
\label{sec:objective}
Given a set $\mathbf{W}$ of $k$ vectors $\mathbf{w}_1, \ldots, \mathbf{w}_k \in
\mathbb{R}^n$, the set of all points in $\mathbb{R}^n$ for which a particular
$\mathbf{w}_i$ is the nearest vector is called the \textit{Voronoi Region}
$\Omega_i$ of this vector, defined as:
\begin{equation}\label{Region}
	\Omega_i := \big\{\mathbf{x} \in \mathbf{R}^n \,\big|\, i = \argmin\limits_{j
\in [1, \ldots, k]}\| \mathbf{x} - \mathbf{w}_j\|\big\}.
\end{equation}

The partition of $\mathbb{R}^n$, or of a manifold $\boldsymbol{\Omega} \subseteq
\mathbb{R}^n$, formed by all the Voronoi regions is called \textit{Voronoi} or
\textit{Dirichlet Tessellation}, and each vector $\mathbf{w}_i$ is referred to
as \textit{Voronoi landmark} or \textit{generator}.

A \textit{Centroidal Voronoi Tessellation} \citep*{du1999centroidal} is a
Voronoi tessellation whose generating points are the centroids (centers of
mass) of the corresponding Voronoi regions, and it minimizes an energy function
$F_{CVT}(\mathbf{W})$, the expected quantization
error\footnote{In some works, like \cite{du1999centroidal}, each point $x$ in
the manifold has a probability value, affecting the integral value. In the work
of \citet{levy2010lpcvt}, however, the function $P(x)$ is not
considered.}, defined as:
\begin{equation}\label{cvt}
	F_{CVT}(\mathbf{W}) = \int_{\boldsymbol{\Omega}}
\left \|\mathbf{x} - \mathbf{w}_{i(\mathbf{x})} \right \|^2 \, d\mathbf{x},
\end{equation}
where $i(\mathbf{x})$ is the index $i$ of the Voronoi landmark nearest to
$\mathbf{x}$.

When the partitioned manifold is a set
$\boldsymbol{\Omega} \subset \mathbb{R}^n$, some of the Voronoi regions will not be
completely contained in $\boldsymbol{\Omega}$, thus leading to the definition of a
\textit{restricted Voronoi cell} as $\Omega_i \cap \boldsymbol{\Omega}$. We can
then decompose the integral in \eqref{cvt} as the sum of the integrals
calculated on each restricted Voronoi region:
\begin{equation}\label{cvtr}
	F_{CVT}(\mathbf{W}) = \sum_i \int_{\Omega_i\cap\boldsymbol{\Omega}}
\left \|\mathbf{x} - \mathbf{w}_i \right \|^2 \, d\mathbf{x},
\end{equation}

$L_p$-Centroidal Voronoi Tessellation \citep[$L_p$-CVT]{levy2010lpcvt} enables
the meshing of a manifold $\boldsymbol{\Omega}$ to be controlled by a given
anisotropy field. It is defined as the minimizer of the
$L_p$-CVT objective function $F_{L_p}$, obtained by injecting an anisotropy
term, i.e. the anisotropy matrix $\mathbf{M_x}$, and the $L_p$-norm into
the standard CVT energy \eqref{cvtr}:
\begin{equation}\label{flp}
	F_{L_p}(\mathbf{W}) = \sum_i \int_{\Omega_i\cap\boldsymbol{\Omega}}
\left \|\mathbf{M_x}(\mathbf{x} - \mathbf{w}_i) \right \|_p^p \, d\mathbf{x},
\end{equation}
here $\|\cdot\|_p$ denotes the $L_p$-norm ($\|\mathbf{V}\|_p = \sqrt[p]{|x|^p
+ |y|^p + |z|^p}$ and $\|\mathbf{V}\|_p^p = |x|^p + |y|^p + |z|^p$). For even
values of $p$, the  $L_p$-norm becomes $\|\mathbf{V}\|_p^p = x^p + y^p + z^p$. We will examine the case in
which the boundary of $\mathbf{\Omega}$ is a piecewise linear complex
(PLC)\footnote{in
\cite{levy2010lpcvt} it is necessary to define the surface as a piecewise linear complex (PLC) for the construction of the tetrahedron in theorem \ref{levy}. It could possible to keep a more general approach defining a new method for the construction of the tetrahedrons, allowing to use generic surfaces instead of PLCs.} $\mathcal{S}$.

In \cite{levy2010lpcvt} the anisotropy field is characterized by a symmetric
tensor field $\mathbf{G_x}$. The \emph{spectral theorem} states that:
\begin{theorem}[Spectral Theorem]\label{specTh} If a $n \times n$ matrix
$\mathbf{A}$ is symmetric, then there is a basis ${\mathbf{e}_1, \ldots, \mathbf{e}_n}$ of
$\mathbb{R}^n$ whose elements are $n$ eigenvectors of $\mathbf{A}$.
\end{theorem}
This means that we can decompose $\mathbf{G_x}$ as:
\begin{equation}\label{spectral}
	\mathbf{G_x} = \mathbf{Q}\mathbf{\Lambda}\mathbf{Q}^t,
\end{equation}
where $\mathbf{Q}$ is a matrix whose columns are the eigenvectors $\mathbf{q}_i$
of $\mathbf{G}$ and $\mathbf{\Lambda}$ is a diagonal matrix whose diagonal
entries are the relative eigenvalues $\lambda_i$.

In \cite{levy2010lpcvt} it is implicitly assumed that $\mathbf{\Lambda}$ is
positive definite, so that it is possible to define $\mathbf{\Sigma}$ and
$\mathbf{M_x}$ as: $$
\mathbf{\Sigma\Sigma}^t = \mathbf{\Lambda},\, \sigma_{ii} =
\sqrt{\lambda_{ii}}\,;\qquad \mathbf{M_x} := \mathbf{\Sigma}^t\mathbf{Q}^t =
\left(\mathbf{Q \Sigma}\right)^t,
$$
and we can rewrite \eqref{spectral} as:
\begin{equation}\label{M_T}
	\mathbf{G_x} = \mathbf{M}_{\mathbf{x}}^t\mathbf{M_x}.
\end{equation}

If we write explicitly $\mathbf{M_x}$ we can see that:
\begin{equation}\label{MT}
	\mathbf{M_x} = \mathbf{\Sigma}^t\mathbf{Q}^t =
\begin{bmatrix}
\sigma_1 & 0 & 0 \\
0 & \sigma_2 & 0 \\
0 & 0 & \sigma_3
\end{bmatrix}
\begin{bmatrix}
q_{1x} & q_{1y} & q_{1z} \\
q_{2x} & q_{2y} & q_{2z} \\
q_{3x} & q_{3y} & q_{3z} \\
\end{bmatrix} =
\begin{bmatrix}
\sigma_1q_{1x} & \sigma_1q_{1y} & \sigma_1q_{1z} \\
\sigma_2q_{2x} & \sigma_2q_{2y} & \sigma_2q_{2z} \\
\sigma_3q_{3x} & \sigma_3q_{3y} & \sigma_3q_{3z} \\
\end{bmatrix} =
\begin{bmatrix}
\left[\sigma_1\mathbf{q}_1\right]^t \\
\left[\sigma_2\mathbf{q}_2\right]^t \\
\left[\sigma_3\mathbf{q}_3\right]^t \\
\end{bmatrix}.
\end{equation}

In order to understand the meaning of the matrix $\mathbf{M_x}$ we will examine
the quadratic form of $\mathbf{G_x(v)}:\mathbb{R}^n \to \mathbb{R}$ applied to a generic vector
$\mathbf{v}$:
\begin{equation}\label{quadratic}
	\mathbf{G_x(v)} = \mathbf{v}^t\mathbf{G_x}\mathbf{v} =
\mathbf{v}^t\mathbf{M_x}^t\mathbf{M_x}\mathbf{v} =
\left(\mathbf{M_x}\mathbf{v}\right)^t\mathbf{M_x}\mathbf{v} =
\left\|\mathbf{M}_{\mathbf{x}}\mathbf{v}\right\|^2.
\end{equation}

If we define the vector $\mathbf{v}_i:= \mathbf{x} -
\mathbf{w}_{i(\mathbf{x})}$, i.e. the distance between a generic point $\mathbf{x}
\in \mathcal{S}$ and its nearest Voronoi landmark inside the corresponding Voronoi
region $\Omega_i$, then \eqref{quadratic} can be interpreted as the squared error and we can
calculate the expected quantization error induced by the Voronoi tessellation,
considering the anisotropy $\mathbf{G_x}$, as:
$$
\sum_i \int_{\Omega_i\cap\boldsymbol{\Omega}}
\left \|\mathbf{M}_{\mathbf{x}}\left(\mathbf{x} - \mathbf{w}_i\right) \right \|_2^2 \,
d\mathbf{x}, $$
which is just $F_{L_2}$, i.e. the $L_2$-based version of \eqref{flp}.

\subsection{The integration $F^T_{L_p}$ over an integration simplex $T$}
\begin{theorem}\label{levy} Given that each closed Voronoi
region in three dimension is a convex polyhedron \citep{okabe2000spatial} and
that the surface $\mathcal{S}$ is a piecewise linear complex, we can divide each region
$\Omega_i \cap \boldsymbol{\Omega}$ into tetrahedrons, i.e. three-dimensional
simplices, and further decompose the integral in \eqref{flp}. Assuming
(a) that each tetrahedron is formed by a Voronoi landmark $\mathbf{w}_i$ and
three vertices\footnote{The subscripts should indicate even that this vertices are
from a particular tetrahedron belonging to the $i$-th Voronoi cell, but for
sake of clarity we will omit those details, leaving them implied.} $\mathbf{C_1,
C_2, C_3}$, (b) that the determinant of the transformation matrix $\mathbf{M}_T$
is $1$, i.e. the transformation doesn't change the volume or energy and (c) that
the index $p$ of the $L_p$-norm is an even number, the integration $F^T_{L_p}$
of the quantization error over each integration simplex $T = T(\mathbf{w}_i,
\mathbf{C_1, C_2, C_3})$, is given by
\citep[see][and appendix \ref{simplexint},
equation \eqref{multi}]{levy2010lpcvt}:
\begin{equation}\label{ftlp}\begin{split}
	F^T_{L_p}&= \int_T \|\mathbf{M}_T(\mathbf{x - w}_i)\|_p^p \, d\mathbf{x}\\
&= \frac{|T|}{\binom{n + p}{n}} \sum_{\alpha + \beta + \gamma = p}
\overline{\mathbf{U}^{*\alpha}_1 *\mathbf{U}^{*\beta}_2 *
\mathbf{U}^{*\gamma}_3},
\end{split}
\end{equation}
\begin{equation*}
	\mbox{where:}\begin{array}{clcl} & \mathbf{U}_j & =
&\quad\mathbf{M}_T(\mathbf{C}_j - \mathbf{w}_i) \\
& \mathbf{V_1 * V_2} & = &\quad[x_1x_2, y_1y_2, z_1z_2]^t \\
& \mathbf{V}^{*\alpha} & = &\quad\mathbf{V * V * \ldots * V}\mbox{  ($\alpha$
times)} \\
& \overline{\mathbf{V}} & = &\quad x + y + z
\end{array}
\end{equation*}
\end{theorem}
\begin{proof}The general rule for integration by substitution, in the
multi-variable case, says that:
\begin{equation}\label{substitution}
	\int_U f(\varphi(\mathbf{x}))\left|\det(\mathbf{J}\varphi)(\mathbf{x})\right| \,
d\mathbf{x} = \int_{\varphi(U)} f(\mathbf{u}) \, d\mathbf{u},
\end{equation}
with the continuously differentiable substitution function $\varphi(\mathbf{x})
= \mathbf{u}$.

If we define:
$$
\varphi(\,\cdot\,) = \mathbf{M}_T(\,\cdot\, - \mathbf{w}_i); \quad f(\,\cdot\,)
=  \left\|\,\cdot\,\right\|_p^p; \quad U = T(\mathbf{w}_i, \mathbf{C}_1,
\mathbf{C}_2, \mathbf{C}_3),$$
then:
$$
f(\varphi(\mathbf{x})) =
\left\|\mathbf{M}_T(\mathbf{x} - \mathbf{w}_i)\right\|_p^p; \quad \varphi(U) =
\varphi(T) = T(0, \mathbf{U}_1, \mathbf{U}_2,
\mathbf{U}_3).
$$

Since $\mathbf{u} = \mathbf{M}_T(\mathbf{x - w}_i)$, the variables identifying
the domain become $\mathbf{u}_{\mathbf{C}_j} = \mathbf{M}_T(\mathbf{C}_j -
\mathbf{w}_i) = \mathbf{U}_j$ and $\mathbf{u}_{\mathbf{w}_i} =
\mathbf{M}_T(\mathbf{w}_i~-~\mathbf{w}_i)~=~0$. We will indicate with
$T^{\prime}$ the simplex $\varphi(T) = T(0, \mathbf{U}_1, \mathbf{U}_2,
\mathbf{U}_3)$ and the variable substitution formula becomes:
$$
\int_T \|\mathbf{M}_T(\mathbf{x - w}_i)\|_p^p
\left|\det\left(J\left(\mathbf{M}_T(\mathbf{x - w}_i)\right)\right)\right| \,
d\mathbf{x} = \int_{T^{\prime}} \left\| \mathbf{u} \right\|_p^p \, d\mathbf{u}
$$

To find the Jacobian matrix we first have to decompose
$\varphi(\mathbf{x})$ in the three components $\left[\varphi(\mathbf{x})_x, \varphi(\mathbf{x})_y,
\varphi(\mathbf{x})_z\right]$. Calling $\mathbf{d}$ the distance $(\mathbf{x}
- \mathbf{w}_i)$ and expanding $\mathbf{M}_T$ as in \eqref{MT} we obtain:
\begin{equation*}
\mathbf{M}_T\mathbf{d} =
\begin{bmatrix}
\left(\sigma_1\mathbf{q}_1\right)^t \\
\left(\sigma_2\mathbf{q}_2\right)^t \\
\left(\sigma_3\mathbf{q}_3\right)^t \\
\end{bmatrix}
\begin{bmatrix}
d_x \\ d_y \\ d_z
\end{bmatrix} =
\begin{bmatrix}
\sigma_1q_{1x}d_x & \sigma_1q_{1y}d_y & \sigma_1q_{1z}d_z \\
\sigma_2q_{2x}d_x & \sigma_2q_{2y}d_y & \sigma_2q_{2z}d_z \\
\sigma_3q_{3x}d_x & \sigma_3q_{3y}d_y & \sigma_3q_{3z}d_z \\
\end{bmatrix}.
\end{equation*}

Being $d_x = x_x - w_{ix}$, each element
$d\varphi_x/dx_x$ of the Jacobian matrix it's found in a manner similar
to:
$$ \frac{d\varphi_x}{dx_x} = \frac{d\left(\sigma_1q_{1x}(x_x - w_{ix}) +
\sigma_1q_{1y}(x_y - w_{iy}) + \sigma_1q_{1z}(x_z - w_{iz})\right)}{dx_x} =
\sigma_1q_{1x}, $$
and the Jacobian matrix of $\varphi(\mathbf{x})$ is then $\mathbf{M}_T$.

Since the determinant of $\mathbf{M}_T$ is $1$ we can express, as in
\cite{levy2010lpcvt}, the energy function as: $$
F^T_{L_p} = \int_T \|\mathbf{M}_T(\mathbf{x - w}_i)\|_p^p \, d\mathbf{x} =
\int_T \|\mathbf{M}_T(\mathbf{x - w}_i)\|_p^p
\left|\det\left(\mathbf{M}_T\right)\right| \, d\mathbf{x} =
\int_{T^{\prime}} \left\|\mathbf{u}\right\|_p^p \, d\mathbf{u}.
$$

Note that:
\begin{equation}\label{upp}
	\left\|\mathbf{u}\right\|_p^p = \left(\sqrt[p]{|x_u|^p + |y_u|^p +
|z_u|^p}\right)^p\ \stackrel{with\, p\, even}{=}\ x_u^p + y_u^p + z_u^p,
\end{equation}
therefore being $\mathbf{u}^{*p} =
\overbrace{\mathbf{u}*\mathbf{u}*\ldots*\mathbf{u}}^{p \mbox{ times}} =
\left[ x_u x_u\ldots x_u, y_u y_u\ldots y_u, z_u z_u\ldots z_u \right]^t =
\left[x_u^p, y_u^p, z_u^p\right]^t$, and being $\overline{\mathbf{u}} = x_u + y_u + z_u$,
the expression becomes:
\begin{equation}\label{up}
	F^T_{L_p} = \int_{T^{\prime}}
\overline{\mathbf{u}^{*p}} \, d\mathbf{u},
\end{equation}
that is a $p$-homogeneous polynomial in $\mathbf{u} = [x_u, y_u,
z_u]^t$. As shown in appendix \ref{simplexint}, we can associate with it a
$p$-linear symmetric form using a polarization formula (see appendix \ref{polformula}):
\begin{align*}
&H(\mathbf{u}\pol{1}, \mathbf{u}\pol{2}, \ldots, \mathbf{u}\pol{p}) =
\frac{1}{p!}\frac{\partial}{\partial \lambda_1}\ldots\frac{\partial}{\partial
\lambda_p} f(\lambda_1\mathbf{u}\pol{1} + \ldots + \lambda_p\mathbf{u}\pol{p})\\
&=\frac{1}{p!}\frac{\partial^p}{\partial\lambda_1\ldots\partial\lambda_p}\overline{{(\lambda_1\mathbf{u}\pol{1}
+ \ldots + \lambda_p\mathbf{u}\pol{p})}^p}\\
&=\frac{1}{p!}\frac{\partial^p}{\partial\lambda_1\ldots\partial\lambda_p}\left[{(\lambda_1
x\pol{1} + \ldots + \lambda_p x\pol{p})}^p + {(\lambda_1 y\pol{1} + \ldots +
\lambda_p y\pol{p})}^p + {(\lambda_1 z\pol{1} + \ldots + \lambda_p z\pol{p})}^p
\right]. \numberthis \label{polar}
\end{align*}

For the sum rule we can solve the derivative for the first parenthetical and
then apply the result to the other two. It's easy to see that
$$\frac{\partial}{\partial\lambda_1}{(\lambda_1 x\pol{1}
+ \ldots + \lambda_p x\pol{p})}^p = px\pol{1}{(\lambda_1 x\pol{1}
+ \ldots + \lambda_p x\pol{p})}^{p - 1},$$ and then
$$\frac{\partial}{\partial\lambda_2}p x\pol{1}{(\lambda_1 x\pol{1}
+ \ldots + \lambda_p x\pol{p})}^{p-1}=p(p-1)x\pol{1}
x\pol{2}{(\lambda_1 x\pol{1} + \ldots + \lambda_p x\pol{p})}^{p - 2},
$$ and so on, until the following result is obtained:
\begin{equation}\label{du} \begin{split}
	\frac{\partial^p}{\partial\lambda_1\ldots\partial\lambda_p}{(\lambda_1
x\pol{1} + \ldots + \lambda_p x\pol{p})}^p &=
p\cdot(p-1)\dotsm(2)\cdot(1)\cdot x\pol{1} x\pol{2} \ldots x\pol{p}{(\lambda_1
x\pol{1} + \ldots)}^0\\ &= p!x\pol{1}x\pol{2}\ldots x\pol{p}.
\end{split} \end{equation}

By putting \eqref{du} back into \eqref{polar}, factoring out $p!$
to cancel the denominator and using \eqref{upp} and \eqref{up}, we
obtain:
\begin{equation}\label{expr1}
	H(\mathbf{u}\pol{1}, \mathbf{u}\pol{2}, \ldots, \mathbf{u}\pol{p})  =
\overline{\mathbf{u}\pol{1} * \mathbf{u}\pol{2} * \dotsm * \mathbf{u}\pol{p}},
\end{equation}
which has to be $p$-linear and symmetric.

The symmetry of \eqref{expr1} is easy to see, since any changes in the order of the arguments of $H$
will only change the order of factors inside the parenthesis, without affecting the result.
On the other hand, $p$-linearity is proved by solving the following equation:
$$
H(\mathbf{u}\pol{1}, \ldots, \lambda\mathbf{u}\pol{i} +
\mu\mathbf{u}\pol{i^{\prime}}, \ldots, \mathbf{u}\pol{p}).
$$

Given that:$$
\lambda\mathbf{u}\pol{i} + \mu\mathbf{u}\pol{i^{\prime}} = \begin{bmatrix}
\lambda x\pol{i} \\ \lambda y\pol{i} \\ \lambda z\pol{i} \end{bmatrix} +
\begin{bmatrix} \mu x\pol{i^{\prime}} \\ \mu y\pol{i^{\prime}} \\
\mu z\pol{i^{\prime}} \end{bmatrix} = \begin{bmatrix}
\lambda x\pol{i} + \mu x\pol{i^{\prime}}\\ \lambda y\pol{i} + \mu
y\pol{i^{\prime}}\\ \lambda z\pol{i} + \mu z\pol{i^{\prime}}
\end{bmatrix},
$$
we can solve for the first operand and then repeat the same procedure twice:
$$
\left(  x^{(1)}  x^{(2)} \dotsm (\lambda x\pol{i} + \mu x\pol{i^{\prime}})\dotsm
x\pol{p} \right) = \lambda( x^{(1)}  x^{(2)} \dotsm x\pol{i}\dotsm x\pol{p}) + \mu
( x^{(1)}  x^{(2)} \dotsm x\pol{i^{\prime}} \dotsm x\pol{p}),
$$
after this, it becomes easy to see that
\begin{align*}
	\lambda( x^{(1)} \dotsm x\pol{i}\dotsm x\pol{p}) &+ \mu
( x^{(1)} \dotsm x\pol{i^{\prime}} \dotsm x\pol{p}) + \lambda( y^{(1)} \dotsm
y\pol{i}\dotsm y\pol{p}) + \mu ( y^{(1)} \dotsm y\pol{i^{\prime}} \dotsm
y\pol{p})\\
&
+ \lambda( z^{(1)} \dotsm z\pol{i}\dotsm
z\pol{p}) + \mu ( z^{(1)} \dotsm z\pol{i^{\prime}} \dotsm z\pol{p})\\
&= \lambda H(\mathbf{u}\pol{1}, \dotsc
, \mathbf{u}\pol{i}, \dotsc,  \mathbf{u}\pol{p}) + \mu H(\mathbf{u}\pol{1},
\dotsc, \mathbf{u}\pol{i^{\prime}}, \dotsc, \mathbf{u}\pol{p}).
\end{align*}

Finally, we can apply theorem \ref{multiint} to equation
\eqref{expr1} to obtain:
\begin{equation} \begin{split}
	F^T_{L_p} = \int_{T^{\prime}}
H(\mathbf{u}, \mathbf{u}, \ldots, \mathbf{u}) \, d\mathbf{u} &=
\frac{|T|}{\binom{n + p}{p}} \sum_{\sum_0^n \alpha_i = p}
\overline{\varphi(\mathbf{w}_i)^{*\alpha_0}
* \varphi(\mathbf{C}_1)^{*\alpha_1} *
\varphi(\mathbf{C}_2)^{*\alpha_2} *
\varphi(\mathbf{C}_3)^{*\alpha_3}}\\
&= \frac{|T|}{\binom{n + p}{p}} \sum_{\alpha + \beta + \gamma = p}
\overline{\mathbf{U}_1^{*\alpha} * \mathbf{U}_2^{*\beta} *
\mathbf{U}_3^{*\gamma}} \qedhere
\end{split} \end{equation}
\end{proof}

\section{Expression of $ \nabla F^T_{L_p} $}
In this section, we will obtain the expression of the overall gradient of $F$ by combining, through the 
chain rule, the gradient of $F$ on the vertices of the integration simplices together with the gradient of $F$ on Voronoi vertices.

The usual chain rule for a function of multiple variables is:
$$
\frac{\partial f(g_1(x), g_2(x), \dotsc , g_n(x))}{\partial x} = \sum_{i=1}^n
\frac{\partial f(g_i(x))}{\partial g_i(x)} \frac{\partial g_i(x)}{\partial x}.
$$
Applying the latter to $F^T_{L_p}$ we obtain:
\begin{equation}\label{chainrule1}
	\frac{\partial F_{L_p}^T (\mathbf{w}_i, \mathbf{C_1, C_2, C_3)})}{\partial\mathbf{W}} =
\frac{\partial F_{L_p}^T}{\partial\mathbf{w}_i}\frac{\partial\mathbf{w}_i}{\partial\mathbf{W}} +
\frac{\partial F_{L_p}^T}{\partial\mathbf{C_1}}\frac{\partial\mathbf{C_1}}{\partial\mathbf{W}} +
\frac{\partial F_{L_p}^T}{\partial\mathbf{C_2}}\frac{\partial\mathbf{C_2}}{\partial\mathbf{W}} +
\frac{\partial F_{L_p}^T}{\partial\mathbf{C_3}}\frac{\partial\mathbf{C_3}}{\partial\mathbf{W}}
\end{equation}
where $\frac{\partial\mathbf{A}}{\partial\mathbf{B}} = {(\frac{\partial a_i}{\partial
b_j})}_{i,j}$ denotes the Jacobian matrix of $\mathbf{A}$ with respect to
$\mathbf{B}$.

Since $F_{L_p}^T$ is a scalar function and $\mathbf{W} \in (\mathbb{R}^n)^k$,
the gradient will also belong to $(\mathbb{R}^n)^k$.

Note that
$\displaystyle\frac{\partial\mathbf{w_i}}{\partial\mathbf{W}} = \left[
\frac{\partial\mathbf{w_i}}{\partial\mathbf{w_1}} \; \cdots
\; \frac{\partial\mathbf{w_i}}{\partial\mathbf{w_k}} \right]
$
and
$$
\frac{\partial\mathbf{w_i}}{\partial\mathbf{w_j}}=
\begin{bmatrix}
\frac{\partial{w_{ix}}}{\partial{w_{jx}}} &
\frac{\partial{w_{ix}}}{\partial{w_{jy}}} &
\frac{\partial{w_{ix}}}{\partial{w_{jz}}} \\ \\
\frac{\partial{w_{iy}}}{\partial{w_{jx}}} &
\frac{\partial{w_{iy}}}{\partial{w_{jy}}} &
\frac{\partial{w_{iy}}}{\partial{w_{jz}}} \\ \\
\frac{\partial{w_{iz}}}{\partial{w_{jx}}} &
\frac{\partial{w_{iz}}}{\partial{w_{jy}}} &
\frac{\partial{w_{iz}}}{\partial{w_{jz}}}
\end{bmatrix}
= \begin{bmatrix}1 & 0 & 0 \\ 0&1&0\\0&0&1 \end{bmatrix}
$$
only if
$i=j$, so that
$\frac{\partial F_{L_p}^T}{\partial\mathbf{w}_i}\frac{\partial\mathbf{w}_i}{\partial\mathbf{W}} =
\frac{\partial F_{L_p}^T}{\partial\mathbf{w}_i}$, meaning that the other elements of
the $n \cdot k$ vector will be $0$.

Following this same reasoning, we see that all other addends are vectors in $\mathbb{R}^n$ which, 
when multiplied each by a $n \times (n \cdot k)$ matrix, produce vectors 
of dimension $n\cdot k$.

We can also express $F^T_{L_p}$ as a function of
$$
\mathbf{U}_{1,2,3} = \begin{bmatrix}\mathbf{U}_1 \\ \mathbf{U}_2 \\ \mathbf{U}_3
\end{bmatrix},
$$
which is a vector in $3 \cdot n$ dimensions.

Putting it all together, we can expand \eqref{chainrule1} as:
\begin{equation}\label{chainrule2}\begin{split}
	\frac{\partial F_{L_p}^T}{\partial\mathbf{W}} =
&\frac{\partial F_{L_p}^T}{\partial\mathbf{U}_{1,2,3}}
\frac{\partial\mathbf{U}_{1,2,3}}{\partial\mathbf{w}_i}
\frac{\partial\mathbf{w}_i}{\partial\mathbf{W}} +
\frac{\partial F_{L_p}^T}{\partial\mathbf{U}_{1,2,3}}
\frac{\partial\mathbf{U}_{1,2,3}}{\partial\mathbf{C_1}}
\frac{\partial\mathbf{C_1}}{\partial\mathbf{W}} +\\
&\frac{\partial F_{L_p}^T}{\partial\mathbf{U}_{1,2,3}}
\frac{\partial\mathbf{U}_{1,2,3}}{\partial\mathbf{C_2}}
\frac{\partial\mathbf{C_2}}{\partial\mathbf{W}} +
\frac{\partial F_{L_p}^T}{\partial\mathbf{U}_{1,2,3}}
\frac{\partial\mathbf{U}_{1,2,3}}{\partial\mathbf{C_3}}
\frac{\partial\mathbf{C_3}}{\partial\mathbf{W}}.
\end{split}
\end{equation}

\subsection{Derivation of $F^T_{L_p}$ relative to the vertices of $T$}
Starting from \eqref{chainrule2} we can now derive each component of the
equation, to obtain the final result.

To solve the above derivative, \cite{levy2010lpcvt} first define:
$$
E^T_{L_p} := \sum_{\alpha +\beta +\gamma = p} \overline{\mathbf{U}_1^{*\alpha} *
\mathbf{U}_2^{*\beta} * \mathbf{U}_3^{*\gamma}}
$$
$$
F^T_{L_p} = \frac{|T|}{\binom{n + p}{p}}E^T_{L_p}
$$
so that:
\begin{equation}\label{df}
	\partial F^T_{L_p} = \partial\left(\frac{|T|}{\binom{n + p}{p}}E^T_{L_p}\right) =
\frac{1}{\binom{n + p}{p}}\left(E^T_{L_p}\partial |T| + |T|\partial
E^T_{L_p}\right)
\end{equation}

The first two steps of the derivation will therefore be the two derivatives
found in \eqref{df}: namely $\displaystyle\frac{\partial
E^T_{L_p}}{\partial\mathbf{U}_{1,2,3}}$ and $\displaystyle\frac{\partial
|T|}{\partial\mathbf{U}_{1,2,3}}$.
\paragraph{First step\\}
We first show the derivation of $\displaystyle\frac{\partial
E^T_{L_p}}{\partial\mathbf{U}_{1,2,3}}$.

Recalling that
\begin{align*}
	\partial E^T_{L_p} &=
\partial \left(
\sum_{\alpha +\beta +\gamma = p}
(x_1^{\alpha}x_2^{\beta}x_3^{\gamma} +
y_1^{\alpha}y_2^{\beta}y_3^{\gamma} +
z_1^{\alpha}z_2^{\beta}z_3^{\gamma})
\right)\\
&= \sum_{\alpha +\beta +\gamma = p}
\Big(\partial(x_1^{\alpha}x_2^{\beta}x_3^{\gamma}) +
\partial(y_1^{\alpha}y_2^{\beta}y_3^{\gamma}) +
\partial(z_1^{\alpha}z_2^{\beta}z_3^{\gamma})\Big),
\end{align*}
we can now examine the derivation
$\displaystyle\frac{\partial E^T_{L_p}}{\partial\mathbf{U}_{1,2,3}}$ step by step:
$$
\frac{\partial E^T_{L_p}}{\partial\mathbf{U}_{1,2,3}} =
\left[\frac{\partial E^T_{L_p}}{\partial\mathbf{U_1}},
\frac{\partial E^T_{L_p}}{\partial\mathbf{U_2}}, \frac{\partial
E^T_{L_p}}{\partial\mathbf{U_3}} \right];
$$
$$
\frac{\partial E^T_{L_p}}{\partial\mathbf{U_1}} =
\left[
\frac{\partial E^T_{L_p}}{\partial x_1},
\frac{\partial E^T_{L_p}}{\partial y_1},
\frac{\partial E^T_{L_p}}{\partial z_1}
\right];
$$
that, for every single combination with $\alpha\geq 1$, equals to:
\begin{align*}
	\left[
\frac{\partial(x_1^{\alpha}x_2^{\beta}x_3^{\gamma})}{\partial x_1} +
\frac{\partial(y_1^{\alpha}y_2^{\beta}y_3^{\gamma})}{\partial x_1} +
\frac{\partial(z_1^{\alpha}z_2^{\beta}z_3^{\gamma})}{\partial x_1},\right.
&\frac{\partial(x_1^{\alpha}x_2^{\beta}x_3^{\gamma})}{\partial y_1} +
\frac{\partial(y_1^{\alpha}y_2^{\beta}y_3^{\gamma})}{\partial y_1} +
\frac{\partial(z_1^{\alpha}z_2^{\beta}z_3^{\gamma})}{\partial y_1},\\
&\frac{\partial(x_1^{\alpha}x_2^{\beta}x_3^{\gamma})}{\partial z_1} +\left.
\frac{\partial(y_1^{\alpha}y_2^{\beta}y_3^{\gamma})}{\partial z_1} +
\frac{\partial(z_1^{\alpha}z_2^{\beta}z_3^{\gamma})}{\partial z_1} \right]\\
&= \left[
\alpha x_1^{\alpha - 1}x_2^{\beta}x_3^{\gamma},
\alpha y_1^{\alpha - 1}y_2^{\beta}y_3^{\gamma},
\alpha z_1^{\alpha - 1}z_2^{\beta}z_3^{\gamma}
\right]\\
&= \alpha \mathbf{U}_1^{*(\alpha - 1)} *
\mathbf{U}_2^{*\beta} * \mathbf{U}_3^{*\gamma}.
\end{align*}

And finally:

$$ \frac{\partial E^T_{L_p}}{\partial\mathbf{U}_{1,2,3}} = {\begin{bmatrix}
\sum\limits_{\alpha +\beta +\gamma = p; \alpha \geq 1} \alpha \mathbf{U}_1^{*(\alpha - 1)} *
\mathbf{U}_2^{*\beta} * \mathbf{U}_3^{*\gamma} \\ \\
\sum\limits_{\alpha +\beta +\gamma = p; \beta \geq 1} \beta
\mathbf{U}_1^{*\alpha} * \mathbf{U}_2^{*(\beta - 1)} * \mathbf{U}_3^{*\gamma} \\ \\
\sum\limits_{\alpha +\beta +\gamma = p; \gamma \geq 1} \gamma
\mathbf{U}_1^{*\alpha} * \mathbf{U}_2^{*\beta} * \mathbf{U}_3^{*(\gamma - 1)} \end{bmatrix}}^t
$$

\paragraph{Second step\\}

Here we will show the derivation of
$\displaystyle\frac{\partial |T|}{\partial\mathbf{U}_{1,2,3}}$.

$T$ is a three-dimensional simplex, i.e. a tetrahedron, having one of its vertices in the origin. This means that $\mathbf{U}_1, \mathbf{U}_2$ and $\mathbf{U}_3$ are all edges of the tetrahedron and therefore its volume is:
$$
|T| = 1/6 \mathbf{U_1 \cdot (U_2 \times U_3)} = 1/6 \mathbf{U_2 \cdot (U_3
\times U_1)} = 1/6 \mathbf{U_3 \cdot (U_1 \times U_2)},
$$
The gradient of the volume with respect to $\mathbf{U}_{1,2,3}$ is:
$$
\frac{\partial |T|}{\partial\mathbf{U}_{1,2,3}} = \frac{1}{6} \Big[
\left[\mathbf{U_2 \times U_3} \right]^t,\,
\left[\mathbf{U_3 \times U_1} \right]^t,\,
\left[\mathbf{U_1 \times U_2} \right]^t \Big].
$$

To prove this,  we will examine the increment $|T|\pol{\varepsilon}$ of $|T|$ with respect to $\mathbf{U}_1$, with $\mathbf{U}_1\pol{\varepsilon} = \mathbf{U}_1 + \varepsilon\mathbf{v}$. From definitions, we know that:
\begin{align*}
	|T|\pol{\varepsilon} &=
\mathbf{U}_1\pol{\varepsilon}\cdot (\mathbf{U}_2 \times \mathbf{U}_3)\\
&= \mathbf{U}_1 \cdot (\mathbf{U}_2 \times \mathbf{U}_3) +
\varepsilon\mathbf{v}\cdot (\mathbf{U}_2 \times \mathbf{U}_3) \\
&= \mathbf{U}_1 \cdot (\mathbf{U}_2 \times \mathbf{U}_3) +
\varepsilon\frac{\partial\left(\mathbf{U}_1 \cdot (\mathbf{U}_2 \times
\mathbf{U}_3)\right)}{\partial\mathbf{U}_1}\mathbf{v} + o(\varepsilon)
\end{align*}
from which
$$
\varepsilon\frac{\partial\left(\mathbf{U}_1 \cdot (\mathbf{U}_2 \times
\mathbf{U}_3)\right)}{\partial\mathbf{U}_1}\mathbf{v} = \varepsilon
(\mathbf{U}_2 \times \mathbf{U}_3)^t \mathbf{v} \Rightarrow
\frac{\partial\left(\mathbf{U}_1 \cdot (\mathbf{U}_2 \times \mathbf{U}_3)
\right)}{\partial\mathbf{U}_1} = (\mathbf{U}_2 \times \mathbf{U}_3)^t, $$
by iterating the same procedure on $\partial\left(\mathbf{U_2 \cdot (U_3
\times U_1)}\right)/\partial\mathbf{U}_2$ and $\partial\left(\mathbf{U_3 \cdot (U_1\times U_2)}\right)/\partial\mathbf{U}_3$ we get to the result.

Intuitively the latter indicates for each vertex a variation proportional to
one sixth of the facet, i.e. a base of the tetrahedron, opposite to the
derivation vertex and with the greatest rate of increase directed along
the perpendicular of that side, i.e. the height.

For completeness and self-containment of this paper, we also report the demonstration of the two-dimensional derivative contained in \cite{levy2010lpcvt} in appendix \ref{dt2d}

\paragraph{Third step\\}
The derivatives of $F^T_{L_p}$ with respect to the vertices of the tetrahedron are:
\begin{equation}\label{dfdt}
	\frac{\partial F^T_{L_p}}{\partial\mathbf{C}} =
\frac{\partial F^T_{L_p}}{\partial\mathbf{U}} \mathbf{M}_T; \qquad
\frac{\partial F^T_{L_p}}{\partial\mathbf{w}_i} = -
\frac{\partial F^T_{L_p}}{\partial\mathbf{C}_1} -
\frac{\partial F^T_{L_p}}{\partial\mathbf{C}_2} -
\frac{\partial F^T_{L_p}}{\partial\mathbf{C}_3}.
\end{equation}

In fact, since
$$
\frac{\partial\left( \mathbf{M}_T(\mathbf{C}_j - \mathbf{w}_i)
\right)}{\partial\mathbf{C}_j} = \mathbf{M}_T \mbox{ and } \frac{\partial\left(
\mathbf{M}_T(\mathbf{C}_j - \mathbf{w}_i) \right)}{\partial\mathbf{w}_i} =
-\mathbf{M}_T, $$
the derivatives of $F^T_{L_p}$ with respect to $\mathbf{w}_i,
\mathbf{C}_1, \mathbf{C}_2$ and$\mathbf{C}_3$ are given, explicitly, by:
$$
\frac{\partial F_{L_p}^T}{\partial\mathbf{w}_i} =
\frac{\partial F_{L_p}^T}{\partial\mathbf{U}_{1,2,3}}
\frac{\partial\mathbf{U}_{1,2,3}}{\partial\mathbf{w}_i} =
\frac{\partial F_{L_p}^T}{\partial\mathbf{U}_{1,2,3}}
\begin{bmatrix}
-\mathbf{M}_T \\ -\mathbf{M}_T \\ -\mathbf{M}_T
\end{bmatrix}
$$
$$
\frac{\partial F^T_{L_p}}{\partial\mathbf{C}_1} =
\frac{\partial F_{L_p}^T}{\partial\mathbf{U}_{1,2,3}}
\frac{\partial\mathbf{U}_{1,2,3}}{\partial\mathbf{C}_1} =
\frac{\partial F_{L_p}^T}{\partial\mathbf{U}_{1,2,3}}
\begin{bmatrix}
\mathbf{M}_T \\ 0 \\ 0
\end{bmatrix} = \frac{\partial F^T_{L_p}}{\partial\mathbf{U}_1} \mathbf{M}_T ;
$$
$$
\frac{\partial F^T_{L_p}}{\partial\mathbf{C}_2} =
\frac{\partial F_{L_p}^T}{\partial\mathbf{U}_{1,2,3}}
\frac{\partial\mathbf{U}_{1,2,3}}{\partial\mathbf{C}_2} =
\frac{\partial F_{L_p}^T}{\partial\mathbf{U}_{1,2,3}}
\begin{bmatrix}
0 \\ \mathbf{M}_T \\ 0
\end{bmatrix} = \frac{\partial F^T_{L_p}}{\partial\mathbf{U}_2} \mathbf{M}_T ;
$$
$$
\frac{\partial F^T_{L_p}}{\partial\mathbf{C}_3} =
\frac{\partial F_{L_p}^T}{\partial\mathbf{U}_{1,2,3}}
\frac{\partial\mathbf{U}_{1,2,3}}{\partial\mathbf{C}_3} =
\frac{\partial F_{L_p}^T}{\partial\mathbf{U}_{1,2,3}}
\begin{bmatrix}
0 \\ 0 \\ \mathbf{M}_T
\end{bmatrix} = \frac{\partial F^T_{L_p}}{\partial\mathbf{U}_3} \mathbf{M}_T ;
$$

\subsection{Gradient of Voronoi vertices}
Given that the simplex vertices, apart from $\mathbf{w}_i$, are Voronoi
vertices, i.e. the intersection of $n + 1$ Voronoi cells\footnote{More
precisely any vertex can be the intersection of \emph{at least} $n + 1$
vertices, but the case in which the number of vertices is more than the
necessary, that we can call the \emph{degenerate} case, is not discussed in
\cite{levy2010lpcvt}, nor will be discussed here.}, each of them is the circumcenter of a Delaunay simplex of vertices $\mathbf{w}_i, \mathbf{w}_j, \mathbf{w}_k, \mathbf{w}_l$, obtained by intersecting the three bisectors (Voronoi edges) between $\mathbf{w}_i$ and $\mathbf{w}_j$, $\mathbf{w}_i$ and $\mathbf{w}_k$, $\mathbf{w}_i$ and $\mathbf{w}_l$. According to \cite{goldman1994intersection} the computation of the point of intersection of three planes can be condensed to:
\begin{equation}\label{intersection}
	\frac{1}{\vec{\mathbf{n}_1}\cdot(\vec{\mathbf{n}_2}\times\vec{\mathbf{n}_3})}\left(
(\mathbf{P}_1\cdot\vec{\mathbf{n}_1})(\vec{\mathbf{n}_2}\times\vec{\mathbf{n}_3})
+
(\mathbf{P}_2\cdot\vec{\mathbf{n}_2})(\vec{\mathbf{n}_3}\times\vec{\mathbf{n}_1})
+
(\mathbf{P}_3\cdot\vec{\mathbf{n}_3})(\vec{\mathbf{n}_1}\times\vec{\mathbf{n}_2})
\right),
\end{equation}
where $\mathbf{P}_j$ is a point on the plane $j$ and $\vec{\mathbf{n}_j}$ is the unit vector normal to the same plane.

If we define each side $\mathbf{L}_j$ of the tetrahedron starting from
$\mathbf{w}_i$ as
\begin{equation}\label{L_j}
	\mathbf{L}_1 = \mathbf{w}_j - \mathbf{w}_i,\quad
\mathbf{L}_2 = \mathbf{w}_k - \mathbf{w}_i,\quad
\mathbf{L}_3 = \mathbf{w}_l - \mathbf{w}_i,
\end{equation}
we can express the points $\mathbf{P}_j$ as follows:
\begin{equation}\label{P_j}
	\mathbf{P}_1 = \frac{\mathbf{w}_j - \mathbf{w}_i}{2} +
	\mathbf{w}_i = \frac{\mathbf{w}_j + \mathbf{w}_i}{2},\quad
	\mathbf{P}_2 = \frac{\mathbf{w}_k + \mathbf{w}_i}{2},\quad
	\mathbf{P}_3 = \frac{\mathbf{w}_l + \mathbf{w}_i}{2}
\end{equation}
and the unit vectors $\vec{\mathbf{n}_j}$ can be expressed in turn as:
\begin{equation}\label{n_j}
	\vec{\mathbf{n}_1} = \frac{\mathbf{L}_1}{\left|\mathbf{L}_1\right|},\quad
	\vec{\mathbf{n}_2} = \frac{\mathbf{L}_2}{\left|\mathbf{L}_2\right|},\quad
	\vec{\mathbf{n}_3} = \frac{\mathbf{L}_3}{\left|\mathbf{L}_3\right|}.
\end{equation}

From \eqref{n_j} we can state that:
\begin{equation}\label{products}
	\vec{\mathbf{n}_j}\times\vec{\mathbf{n}_k} =
	\frac{\mathbf{L}_j\times\mathbf{L}_k}{\left|\mathbf{L}_j\right|\left|\mathbf{L}_k\right|}
\mbox{ and }
	\vec{\mathbf{n}_j}\cdot(\vec{\mathbf{n}_k}\times\vec{\mathbf{n}_l}) =
	\frac{\mathbf{L}_j\cdot\left(\mathbf{L}_k\times\mathbf{L}_l\right)}
	{\left|\mathbf{L}_j\right|\left|\mathbf{L}_k\right|\left|\mathbf{L}_l\right|}.
\end{equation}

Eventually, substituting \eqref{P_j}, \eqref{n_j} and \eqref{products} in
\eqref{intersection}, we obtain the condensed equation for each vertex
$\mathbf{C}$, defined by the Voronoi generators $\mathbf{w}_i, \mathbf{w}_{j},
\mathbf{w}_{k}, \mathbf{w}_{l}$:
\begin{equation}\label{tetinters}
	\begin{split}
	\mathbf{C} =
\frac{\left|\mathbf{L}_{1}\right|\left|\mathbf{L}_{2}\right|\left|\mathbf{L}_{3}\right|}
{\mathbf{L}_{1}\cdot\left(\mathbf{L}_{2}\times\mathbf{L}_{3}\right)}
&\left[\;
\left(\frac{\mathbf{w}_{j} + \mathbf{w}_i}{2} \cdot
\frac{\mathbf{L}_{1}}{\left|\mathbf{L}_{1}\right|}\right)
\frac{\mathbf{L}_{2}\times\mathbf{L}_{3}}
{\left|\mathbf{L}_{2}\right|\left|\mathbf{L}_{3}\right|} \right.\\
&+ \left(\frac{\mathbf{w}_{k} + \mathbf{w}_i}{2} \cdot
\frac{\mathbf{L}_{2}}{\left|\mathbf{L}_{2}\right|}\right)
\frac{\mathbf{L}_{3}\times\mathbf{L}_{1}}
{\left|\mathbf{L}_{3}\right|\left|\mathbf{L}_{1}\right|} \\
&+ \left.\left(\frac{\mathbf{w}_{l} + \mathbf{w}_i}{2} \cdot
\frac{\mathbf{L}_{3}}{\left|\mathbf{L}_{3}\right|}\right)
\frac{\mathbf{L}_{1}\times\mathbf{L}_{2}}
{\left|\mathbf{L}_{1}\right|\left|\mathbf{L}_{2}\right|}
\right].
\end{split}
\end{equation}

By solving in \eqref{tetinters} the first addend alone, we obtain:
\begin{align*}
	\frac{\left|\mathbf{L}_{1}\right|\left|\mathbf{L}_{2}\right|\left|\mathbf{L}_{3}\right|}
{\mathbf{L}_{1}\cdot\left(\mathbf{L}_{2}\times\mathbf{L}_{3}\right)}
\left(\frac{\mathbf{w}_{j} + \mathbf{w}_i}{2} \cdot
\frac{\mathbf{L}_{1}}{\left|\mathbf{L}_{1}\right|}\right)
\frac{\mathbf{L}_{2}\times\mathbf{L}_{3}}
{\left|\mathbf{L}_{2}\right|\left|\mathbf{L}_{3}\right|} &=
\frac{1}{\mathbf{L}_{1}\cdot\left(\mathbf{L}_{2}\times\mathbf{L}_{3}\right)}
\frac{(\mathbf{w}_{j} + \mathbf{w}_i)(\mathbf{w}_{j} -
\mathbf{w}_i)}{2}(\mathbf{L}_{2}\times\mathbf{L}_{3}) \\
&= \frac{1} {\mathbf{L}_{1}\cdot\left(\mathbf{L}_{2}\times\mathbf{L}_{3}\right)}
(\mathbf{L}_{2}\times\mathbf{L}_{3})
\frac{\mathbf{w}_{j}^2 - \mathbf{w}_i^2}{2}.
\end{align*}

Then, by using the same result for the other two addends, we can write:
\begin{equation}\label{Cvertex}
	\begin{split}
	\mathbf{C} =
\frac{1}{\mathbf{L}_{1}\cdot\left(\mathbf{L}_{2}\times\mathbf{L}_{3}\right)}
&\left[\;
\left(\mathbf{L}_{2}\times\mathbf{L}_{3}\right)
\left(\frac{1}{2}(\mathbf{w}_{j}^2 - \mathbf{w}_i^2)\right) \right.\\
&+ \left(\mathbf{L}_{3}\times\mathbf{L}_{1}\right)
\left(\frac{1}{2}(\mathbf{w}_{k}^2 - \mathbf{w}_i^2)\right)\\
&+ \left.\left(\mathbf{L}_1\times\mathbf{L}_{2}\right)
\left(\frac{1}{2}(\mathbf{w}_{l}^2 - \mathbf{w}_i^2)\right)
\right].
\end{split}
\end{equation}

If we consider a matrix $\mathbf{A}$ having $\mathbf{L}_{1}^t,
\mathbf{L}_{2}^t$ and $\mathbf{L}_{3}^t$ as its rows, the determinant is:
$\det\mathbf{A} = \mathbf{L}_{1}\cdot\left(\mathbf{L}_{2}\times\mathbf{L}_{3}\right)$,
and it follows from the definitions of inverse matrix and of cross product that:
$$
\mathbf{A}^{-1} =
\frac{1}{\mathbf{L}_{1}\cdot\left(\mathbf{L}_{2}\times\mathbf{L}_{3}\right)}
\begin{bmatrix}
\left[\mathbf{L}_{2}\times\mathbf{L}_{3}\right]^t \\
\left[\mathbf{L}_{3}\times\mathbf{L}_{1}\right]^t \\
\left[\mathbf{L}_{1}\times\mathbf{L}_{2}\right]^t
\end{bmatrix}.
$$

It is now easy to see that $\mathbf{C}$ may be found by:
$$
\mathbf{C} = \mathbf{A}^{-1}\mathbf{B},
\mbox{ where }
\mathbf{A} = \begin{bmatrix}
[\mathbf{w}_{j} - \mathbf{w}_i]^t \\ \\
[\mathbf{w}_{k} - \mathbf{w}_i]^t \\ \\
[\mathbf{w}_{l} - \mathbf{w}_i]^t
\end{bmatrix} ; \mathbf{B} =
\frac{1}{2}\begin{bmatrix}
\mathbf{w}_{j}^2 - \mathbf{w}_i^2 \\ \\
\mathbf{w}_{k}^2 - \mathbf{w}_i^2 \\ \\
\mathbf{w}_{l}^2 - \mathbf{w}_i^2
\end{bmatrix}.
$$

\paragraph{Fourth step\\}
To calculate $\displaystyle\frac{\partial\mathbf{C}}{\partial\mathbf{W}}$ we first need to recall some matrix derivation rules \citep{minka1997old}:
\begin{equation}\label{mat1}
	\partial(\mathbf{AB}) = \partial\mathbf{(A)B} + \mathbf{A}\partial\mathbf{(B)}
\end{equation}
\begin{equation}\label{mat2}
	\partial(\mathbf{A}^{-1}) = - \mathbf{A}^{-1}(\partial\mathbf{A})\mathbf{A}^{-1}.
\end{equation}

Using first \eqref{mat1} then \eqref{mat2}, the expression of $\partial\mathbf{C}$ can
be expanded:
\begin{align*}
	\partial\mathbf{C} = \partial(\mathbf{A}^{-1}\mathbf{B}) &= \partial(\mathbf{A}^{-1})\mathbf{B} +
\mathbf{A}^{-1}\partial(\mathbf{B}) \\
&= - \mathbf{A}^{-1}\partial(\mathbf{A})\mathbf{A}^{-1}\mathbf{B} +
\mathbf{A}^{-1}\partial(\mathbf{B}) \\
&= \mathbf{A}^{-1}\left(\partial(\mathbf{B}) - \partial(\mathbf{A})\mathbf{C}\right).
\end{align*}

Both $\partial\mathbf{B}/\partial\mathbf{W}$ and
$\partial\mathbf{A}/\partial\mathbf{W}$ depend only from $\mathbf{w}_i, \mathbf{w}_{j}, \mathbf{w}_{k}$ and $\mathbf{w}_{l}$,
just like the circumcenter $\mathbf{C}$, so we can substitute
$\partial(\mathbf{B})$ and $\partial(\mathbf{A})\mathbf{C}$ with:
$$
\partial\mathbf{B} =
\frac{1}{2}\begin{bmatrix}
\frac{\partial\mathbf{w}_{j}^2 - \mathbf{w}_i^2}{\partial\mathbf{w}_i} &
\frac{\partial\mathbf{w}_{j}^2 - \mathbf{w}_i^2}{\partial\mathbf{w}_{j}} &
\frac{\partial\mathbf{w}_{j}^2 - \mathbf{w}_i^2}{\partial\mathbf{w}_{k}} &
\frac{\partial\mathbf{w}_{j}^2 - \mathbf{w}_i^2}{\partial\mathbf{w}_{l}} \\ \\
\frac{\partial\mathbf{w}_{k}^2 - \mathbf{w}_i^2}{\partial\mathbf{w}_i} &
\frac{\partial\mathbf{w}_{k}^2 - \mathbf{w}_i^2}{\partial\mathbf{w}_{j}} &
\frac{\partial\mathbf{w}_{k}^2 - \mathbf{w}_i^2}{\partial\mathbf{w}_{k}} &
\frac{\partial\mathbf{w}_{k}^2 - \mathbf{w}_i^2}{\partial\mathbf{w}_{l}} \\ \\
\frac{\partial\mathbf{w}_{l}^2 - \mathbf{w}_i^2}{\partial\mathbf{w}_i} &
\frac{\partial\mathbf{w}_{l}^2 - \mathbf{w}_i^2}{\partial\mathbf{w}_{j}} &
\frac{\partial\mathbf{w}_{l}^2 - \mathbf{w}_i^2}{\partial\mathbf{w}_{k}} &
\frac{\partial\mathbf{w}_{l}^2 - \mathbf{w}_i^2}{\partial\mathbf{w}_{l}}
\end{bmatrix} =
\begin{bmatrix} -\mathbf{w}_i^t & \mathbf{w}_{j}^t & 0 & 0 \\ \\
-\mathbf{w}_i^t & 0 & \mathbf{w}_{k}^t & 0 \\ \\ -\mathbf{w}_i^t & 0 & 0 &
\mathbf{w}_{l}^t \end{bmatrix};
$$
$$
\partial(\mathbf{A})\mathbf{C} =
\begin{bmatrix}
\frac{\partial[\mathbf{w}_{j} - \mathbf{w}_i]^t}{\partial\mathbf{w}_i} &
\frac{\partial[\mathbf{w}_{j} - \mathbf{w}_i]^t}{\partial\mathbf{w}_{j}} &
\frac{\partial[\mathbf{w}_{j} - \mathbf{w}_i]^t}{\partial\mathbf{w}_{k}} &
\frac{\partial[\mathbf{w}_{j} - \mathbf{w}_i]^t}{\partial\mathbf{w}_{l}} \\ \\
\frac{\partial[\mathbf{w}_{k} - \mathbf{w}_i]^t}{\partial\mathbf{w}_i} &
\frac{\partial[\mathbf{w}_{k} - \mathbf{w}_i]^t}{\partial\mathbf{w}_{j}} &
\frac{\partial[\mathbf{w}_{k} - \mathbf{w}_i]^t}{\partial\mathbf{w}_{k}} &
\frac{\partial[\mathbf{w}_{k} - \mathbf{w}_i]^t}{\partial\mathbf{w}_{l}} \\ \\
\frac{\partial[\mathbf{w}_{l} - \mathbf{w}_i]^t}{\partial\mathbf{w}_i} &
\frac{\partial[\mathbf{w}_{l} - \mathbf{w}_i]^t}{\partial\mathbf{w}_{j}} &
\frac{\partial[\mathbf{w}_{l} - \mathbf{w}_i]^t}{\partial\mathbf{w}_{k}} &
\frac{\partial[\mathbf{w}_{l} - \mathbf{w}_i]^t}{\partial\mathbf{w}_{l}}
\end{bmatrix}\left(\mathbf{C}\right),
$$
that if $\frac{\partial[\mathbf{w}_{j} - \mathbf{w}_i]^t}{\partial\mathbf{w}_i}
= \begin{bmatrix} 1&0&0 & 0&1&0 & 0&0&1 \end{bmatrix}$ (that can be reasonable,
since $[\mathbf{w}_{j} - \mathbf{w}_i]^t$ is a row vector) and $\partial\mathbf{A}$
is treated as a block matrix leads to the result: $$
\partial(\mathbf{A})\mathbf{C} =
\begin{bmatrix}
- \mathbf{C}^t & \mathbf{C}^t & 0 & 0 \\ \\
- \mathbf{C}^t & 0 & \mathbf{C}^t & 0 \\ \\
- \mathbf{C}^t & 0 & 0 & \mathbf{C}^t
\end{bmatrix},
$$
getting to the result:
\begin{equation}
	\frac{\partial\mathbf{C}}{\partial\mathbf{W}} =
{\begin{bmatrix}
[\mathbf{w}_{j} - \mathbf{w}_i]^t \\ \\
[\mathbf{w}_{k} - \mathbf{w}_i]^t \\ \\
[\mathbf{w}_{l} - \mathbf{w}_i]^t
\end{bmatrix}}^{-1}
\begin{bmatrix}
\left[\mathbf{C} - \mathbf{w}_i\right]^t &
\left[\mathbf{w}_{j} - \mathbf{C}\right]^t & 0 & 0 \\ \\
\left[\mathbf{C} - \mathbf{w}_i\right]^t &
0 & \left[\mathbf{w}_{k} - \mathbf{C}\right]^t & 0 \\ \\
\left[\mathbf{C} - \mathbf{w}_i\right]^t & 0 & 0 &
\left[\mathbf{w}_{l} - \mathbf{C}\right]^t
\end{bmatrix},
\end{equation}
being the other elements of $\left(\partial(\mathbf{B}) -
\partial(\mathbf{A})\mathbf{C}\right)$ all zeroes.

\section*{Conclusions}\label{sec:conclusions}

The complete, expanded derivations presented here for the objective function of $L_p$-CVT and its gradient
help highlighting the conditions of applicability of the anisotropy field, as stated in theorem \ref{levy}. 
As a point of relevance, this shows the possibility of combining the mathematical framework of \cite{levy2010lpcvt}
with the well-known computational method of the Local Principal Component  Analysis (LPCA)
\citep{kambhatla1997dimension}; in fact, the LPCA tensor can become the anisotropy term in the 
framework presented as long as  the simple conditions stated in theorem \ref{levy} are fulfilled.
Hopefully, the present work will be the base for further studies on the application of Lp-CVT to surface reconstruction 
and remeshing from noisy point clouds.

In passing, these same derivations show a minor inaccuracy in the original derivation of
\cite{levy2010lpcvt}, with respect to surface meshing (see appendix \ref{dt2d}).

\section*{Acknowledgments}

The authors want to thank Epifanio Virga for his substantial help with the details of the derivation. Needless to say, all remaining errors and inaccuracies are entirely our responsibility.

\appendix

%

\section{Integrating a $p$-homogeneous polynomial over an $n$-dimensional
simplex}\label{simplexint}
Considering $n + 1$ points, or \emph{vertices}, $x_0, x_1, \ldots, x_n$, such
that the vectors $(x_i~-~x_0)$ are linearly independent, the $n$-dimensional
non-degenerate\footnote{The term \textit{degenerate} indicate a
simplex in which the edges starting from a given vertex aren't linearly
independent, i.e. a $n$-dimensional simplex that has intrinsic dimension $d<n$.
} simplex $\Delta_n\subset \mathbb{R}^n$ is the set of $x$ such
that $x\in\Delta_n$ if and only if $x$ is a convex combination $\sum_0^n
\lambda_{i}x_i$, with $\lambda_i \ge 0$ and $\sum_0^i \lambda_i =
1$.

We are interested in computing
$$
\int_{\Delta_n} q(x) \, dx,
$$
where $\Delta_n$ is an $n$-dimensional simplex as described above, and $q(x):
\mathbb{R}^n \to \mathbb{R}$ is a real \textit{p-homogeneous} polynomial, i.e. $q(\lambda x)=\lambda^p q(x)$ for all
$\lambda > 0$, $x\in \mathbb{R}^n$ and some integer $p \geq 0$.

We first need to introduce some concepts and notation. With every symmetric
$p$-linear form\footnote{Symmetric
$p$-linear means that the polynomial is function of $p$ variables, is invariant
in respect to the order of its variables, i.e. $H(\mathbf{x}_1, \mathbf{x}_2,
\ldots) = H(\mathbf{x}_2, \mathbf{x}_1, \ldots)$, and it's linear in each of them,
i.e. $H(\mathbf{x}_1, \ldots, \lambda\mathbf{x}_i + \mu\mathbf{x}^{\prime}_i,
\ldots, \mathbf{x}_p) = \lambda H(\mathbf{x}_1, \ldots, \mathbf{x}_i, \ldots
\mathbf{x}_p) + \mu H(\mathbf{x}_1, \ldots, \mathbf{x}^{\prime}_i, \ldots,
\mathbf{x}_p), \forall 1 \le i \le p$.} $H : (\mathbb{R}^n)^p \to \mathbb{R}$,
given by:
$$
(x_1, x_2, \ldots, x_p) \mapsto H(x_1,x_2,\ldots, x_p), \qquad\qquad x_1, x_2,
\ldots, x_p \in \mathbb{R}^n,
$$
one may associate a $p$-homogeneous polynomial $x \mapsto f(x) :=
H(\overbrace{x,x,\ldots,x}^{p \mbox{ times}})$ and conversely, using a
polarization formula (see appendix \ref{polformula}),with every $p$-homogeneous
polynomial $f : \mathbb{R}^n \to \mathbb{R}$ one may associate\footnote{Kirwan and Ryan, in
\cite{kirwan1998extendibility} state that if $f(x) \in \mathcal{P}(^pX)$, where
$\mathcal{P}(^pX)$ denotes the Banach space of bounded, $p$-homogeneous
polynomials from $X$ into $\mathbb{R}$, then there exists a unique bounded
symmetric multi-linear mapping $H : X^n \to \mathbb{R}$ such that $f(x) = H(x;
\ldots ; x), \forall x \in X$. This result is valid, in our case, since we work
in the Euclidean space, that is a banach space.} a symmetric $p$-linear form $H
: (\mathbb{R}^n)^p \to \mathbb{R}$, such that $H(\overbrace{x,x,\ldots,x}^{p
\mbox{ times}}) = f(x)$.

Therefore, we now consider the integration of a $p$-linear form $H$ over the
simplex $\Delta_n$ 	\citep{lasserre2001multi2}.

\begin{theorem}[Lasserre, Avrachenkov]\label{multiint} Let $x_0, x_1, \ldots,
x_n$ be the $n + 1$ vertices of an $n$-dimensional simplex $\Delta_n$. Then, for a symmetric
$p$-linear form $H : (\mathbb{R}^n)^p \to \mathbb{R}$, one has
\begin{equation} \label{multi}
	\int_{\Delta_n} H(x, x, \ldots, x) \, dx = \frac{Vol(\Delta_n)}{\binom{n +
p}{p}} \left[ \sum_{\sum_0^n \alpha_i = p} H(x_0^{\alpha_0}, x_1^{\alpha_1}, \ldots,
x_n^{\alpha_n}) \right]
\end{equation}
where the notation $H(x_0^{\alpha_0}, x_1^{\alpha_1}, \ldots,
x_n^{\alpha_n})$ means that $x_0$ appears $\alpha_0$ times, $x_1$ appears
$\alpha_1$ times, \ldots, $x_n$ appears $\alpha_n$ times.
\end{theorem}

The number of possible $p$-multiset of $n + 1$ objects, that is, the $n + 1$
vertices of $\Delta_n$ as the $p$ variables of $H$, is the same of the simple
$p$-combinations, i.e. without repetitions, of $(n + 1) + p - 1$ objects, that is given by the binomial
coefficient $\binom{(n + 1) + p - 1}{p} = \binom{n + p}{p}$, as stated in
(\ref{multi}).

Since every polynomial can be represented as a sum of
homogeneous polynomials, (\ref{multi}) can be easily applied to integrate an arbitrary polynomial over a
simplex.

\section{Polarization formula}\label{polformula}
Let $f(\mathbf{u})$ be a
polynomial of $n$ variables $\mathbf{u} = (u_1, u_2, \ldots u_n)$, $f$
homogeneous of degree $p$. Let $\mathbf{u}^{(1)}, \mathbf{u}^{(2)}, \ldots,
\mathbf{u}^{(p)}$ be a collection of indeterminates with $\mathbf{u}^{(i)} =
(u_1^{(i)}, u_2^{(i)}, \ldots , u_n^{(i)})$. The \textit{polar form} of $f$ is
then a polynomial in $pn$ variables
$$
F(\mathbf{u}^{(1)}, \mathbf{u}^{(2)}, \ldots, \mathbf{u}^{(p)}),
$$
which is linear in each $\mathbf{u}^{(i)}$, symmetric among the
$\mathbf{u}^{(i)}$, and such that $F(\mathbf{u}, \mathbf{u}, \ldots,
\mathbf{u}) = f(\mathbf{u})$.

The polar form of $f$ is given by the following construction:
$$
F(\mathbf{u}^{(1)}, \mathbf{u}^{(2)}, \ldots, \mathbf{u}^{(d)}) = \frac{1}{d!}
\frac{\partial}{\partial \lambda_1}\ldots \frac{\partial}{\partial \lambda_d}
f\left(\lambda_1\mathbf{u}^{(1)} + \lambda_2\mathbf{u}^{(2)} + \ldots +
\lambda_d\mathbf{u}^{(d)}\right).
$$

For example, if $\mathbf{u} = (x, y)$ and $f(\mathbf{u}) = x^2 + 3xy + 2y^2$,
the polarization of $f$ is a function in $\mathbf{u}^{(1)} = ( x^{(1)},  y^{(1)} )$
and $\mathbf{u}^{(2)} = ( x^{(2)} ,  y^{(2)})$, with $\lambda_1\mathbf{u}^{(1)} +
\lambda_2\mathbf{u}^{(2)} = \lambda_1\begin{bmatrix}  x^{(1)} \\  y^{(2)}
\end{bmatrix} + \lambda_2\begin{bmatrix}  x^{(2)} \\  y^{(2)} \end{bmatrix}$,
given by:
\begin{align*}
	&F(\mathbf{u}^{(1)},\mathbf{u}^{(2)}) = \frac{1}{2!}\frac{\partial}{\partial
\lambda_1}\frac{\partial}{\partial \lambda_2}f(\lambda_1\mathbf{u}^{(1)} +
\lambda_2\mathbf{u}^{(2)})\\
&= \frac{1}{2}\frac{\partial}{\partial
\lambda_1}\frac{\partial}{\partial \lambda_2}\left({(\lambda_1 x^{(1)} +
\lambda_2 x^{(2)})}^2 + 3(\lambda_1 x^{(1)} + \lambda_2 x^{(2)})(\lambda_1 y^{(1)} +
\lambda_2 y^{(2)}) + 2{(\lambda_1 y^{(1)} + \lambda_2 y^{(2)})}^2\right)\\
&= \frac{1}{2}\frac{\partial}{\partial\lambda_2}\left(2 x^{(1)}(\lambda_1 x^{(1)} +
\lambda2 x^{(2)}) + 3(2\lambda_1 x^{(1)} y^{(1)} + \lambda_2 x^{(2)} y^{(1)} +
\lambda_2 x^{(1)} y^{(2)}) + 4 y^{(1)}(\lambda_1 y^{(1)} + \lambda_2 y^{(2)})\right)\\
&=\frac{1}{2}(2 x^{(1)} x^{(2)} + 3 x^{(2)} y^{(1)} + 3 x^{(1)} y^{(2)} + 4 y^{(1)} y^{(2)})\\
&=  x^{(1)}  x^{(2)} +
\frac{3}{2} x^{(1)}  y^{(2)} + \frac{3}{2} x^{(2)} y^{(1)} + 2 y^{(1)}  y^{(2)}
\end{align*}

\section{Derivation of a 2D simplex with respect to
$\displaystyle\mathbf{U}_{1,2,3}$}\label{dt2d}
If $T$ is a surface simplex, i.e. a triangle,
his area is $|T| = 1/2 ||\mathbf{N}||$, with $$
\mathbf{N} = (\mathbf{U_1} - \mathbf{U_3}) \times
(\mathbf{U_2} - \mathbf{U_3}) = (\mathbf{U_2} - \mathbf{U_1}) \times
(\mathbf{U_3} - \mathbf{U_1}) = (\mathbf{U_3} - \mathbf{U_2}) \times
(\mathbf{U_1} - \mathbf{U_2}).
$$

We have to find the gradient
$\displaystyle\frac{\partial|T|}{\partial\mathbf{U}_{1,2,3}}$ with respect to
the vector composed by $\mathbf{U}_1, \mathbf{U}_2$ and $\mathbf{U}_3$.

The first three components of the resultant vector, i.e.
$\displaystyle\frac{\partial|T|}{\partial\mathbf{U}_1}$ are
obtained with the following method:
\begin{equation}\begin{split}
	\frac{\partial|T|}{\partial\mathbf{U}_1} =
\frac{\partial1/2||\mathbf{N}||}{\partial\mathbf{U}_1} &=
\frac{1}{2}{\frac{\partial\left({\mathbf{N \cdot
N}}\right)^{1/2}}{\partial\mathbf{U}_1}} \\ &=
\frac{1}{2}\frac{1}{2}{\left(\mathbf{N \cdot N}\right)}^{-1/2}
{\frac{\partial\left(\mathbf{N \cdot N}\right)}{\partial\mathbf{U}_1}} \\
&= \frac{1}{4}\frac{1}{\left(\mathbf{N \cdot
N}\right)^{1/2}}2\frac{\partial\mathbf{N}}{\partial\mathbf{U}_1}\mathbf{N} \\
&=\frac{1}{4}\frac{1}{|T|}\frac{\partial\mathbf{N}}
{\partial\mathbf{U}_1}\mathbf{N}. \label{dT1}
\end{split}\end{equation}

We can now derive
$\displaystyle\frac{\partial\mathbf{N}}{\partial\mathbf{U}_1}$. We will examine
the increment $\mathbf{N}\pol{\varepsilon}$ of $\mathbf{N}$ with respect to
$\mathbf{U}_1$, with $\mathbf{U}_1\pol{\varepsilon} = \mathbf{U}_1 +
\varepsilon\mathbf{v}$. We know from definition that:
\begin{equation}
	\mathbf{N}\pol{\varepsilon} = (\mathbf{U_1}\pol{\varepsilon} - \mathbf{U_3})
\times (\mathbf{U_2} - \mathbf{U_3}) = \mathbf{N} +
\varepsilon\frac{\partial\mathbf{N}} {\partial\mathbf{U}_1}\mathbf{v} +
o(\varepsilon),
\end{equation}
and expanding $\mathbf{N}\pol{\varepsilon}$ we obtain:
\begin{equation}\begin{split}
	\mathbf{N}\pol{\varepsilon} &=
(\mathbf{U_1}\pol{\varepsilon} - \mathbf{U_3}) \times (\mathbf{U_2} - \mathbf{U_3})
\\ &= (\mathbf{U_1} + \varepsilon\mathbf{v} - \mathbf{U_3}) \times (\mathbf{U_2}
- \mathbf{U_3})
\\ &= (\mathbf{U_1} - \mathbf{U_3}) \times (\mathbf{U_2} - \mathbf{U_3}) +
(\varepsilon\mathbf{v}) \times (\mathbf{U_2} - \mathbf{U_3})
\\ &= \mathbf{N} + \varepsilon\mathbf{v}\times({\mathbf{U}_2 - \mathbf{U}_3})
\\ &= \mathbf{N} - \varepsilon({\mathbf{U}_2 -
\mathbf{U}_3})\times\mathbf{v}.
\end{split}
\end{equation}

Any cross product between two vectors can be expressed with an antisymmetric
tensor, associated to the first vector, applied to the second one, such that
\begin{equation}\label{antysim}
	\mathbf{w} \times \mathbf{v} = \mathbf{W}\mathbf{v}, \mbox{ and }
\mathbf{W}^t\mathbf{v} = - \mathbf{w} \times \mathbf{v} = \mathbf{v} \times
\mathbf{w}, \quad \forall \ \mathbf{v}\! \in \mathbb{R}^n
\end{equation}
therefore
\begin{equation}\label{W}
	\left(\frac{\partial\mathbf{N}}{\partial\mathbf{U}_1}\right)\mathbf{v}
= \mathbf{W}\mathbf{v} = - ({\mathbf{U}_2 -
\mathbf{U}_3})\times\mathbf{v}.
\end{equation}

We need only the value of $\frac{\partial\mathbf{N}} {\partial\mathbf{U}_1}$
applied to e generic vector $v$, to find the value of \eqref{dT1}.
Substituting \eqref{W} in \eqref{dT1} we obtain:
\begin{equation}\label{dT1final}
	\left(\frac{\partial|T|}{\partial\mathbf{U}_1}\right) =
\frac{-1}{4|T|}\left[ \mathbf{N}\times(\mathbf{U}_2 -
\mathbf{U}_3)\right]^t.
\end{equation}

To obtain the entire gradient, it's enough to obtain the equations equivalent
to \eqref{dT1final} for $\frac{\partial}{\partial\mathbf{U}_2}$ and
$\frac{\partial}{\partial\mathbf{U}_3}$.

This is the result:
$$ \frac{\partial|T|}{\partial\mathbf{U}_{1,2,3}} =
\frac{-1}{4|T|}
\Big[
\left[\mathbf{N \times (U_2 - U_3)} \right]^t,\,
\left[\mathbf{N \times (U_3 - U_1)} \right]^t,\,
\left[\mathbf{N \times (U_1 - U_2)} \right]^t
\Big].
$$

Intuitively this result indicates for each vertex a variation proportional to
half of the opposite side of the triangle (base) and with the greatest rate of
increase in the orthogonal direction with respect to that side (height).

\bibliographystyle{plainnat}
\bibliography{ParigiPiastra}

\end{document}